\documentclass[onecolumn]{ifacconf}
\usepackage{comment}
\usepackage{graphicx}      
\usepackage{natbib}        
\usepackage{amsmath,amssymb,amsfonts}
\usepackage{svg}
\usepackage{xcolor}
\usepackage{bm}
\usepackage[utf8]{inputenc}
\usepackage{todonotes}
\usepackage[T1]{fontenc}
\usepackage[hyphens]{url}
\usepackage[dvipsnames]{xcolor}

\usepackage{subcaption}

\newtheorem{assumption}{Assumption}

\newtheorem{definition}{Definition}
\newcommand{\col}{\text{col}}
\newcommand{\G}{\mathcal{G}}

\begin{document}
\begin{frontmatter}

\title{Hierarchical parameter estimation for distributed networked systems:\\a dynamic consensus approach}

\thanks[footnoteinfo]{This work was supported in part by the Secretaría de Ciencia, Humanidades, Tecnología e Innovación (SECIHTI), México, previously administered by the Consejo Nacional de Humanidades, Ciencias y Tecnologías (CONAHCYT) with grant number 1229622, and in part by projects PID2021-124137OB-I00 and 
PID2024-159279OB-I00 funded by MICIU/AEI/10.13039/501100011033 and by ERDF/EU, by  
project REMAIN S1/1.1/E0111 (Interreg Sudoe Programme, ERDF), and via project DGA T45\_23R (Gobierno de Aragón). Grant reference BG24/00121 funded by MICIU/AEI/10.13039/501100011033. {\color{red}This is the pre-peer reviewed version of the following article: Méndez-Castillo, A., Aldana-López, R., Ramírez-Treviño A., Aragüés, R., Gómez-Gutierrez D., "Hierarchical parameter estimation for distributed networked systems: A dynamic consensus approach", Nonlinear Analysis: Hybrid Systems, 2026, which has been published in final form at https://doi.org/10.1016/j.nahs.2026.101730. This article may be used for non-commercial purposes in accordance with Wiley Terms and Conditions for Use of Self-Archived Versions. \textbf{Please cite the publisher’s version.} For the publisher’s version and full citation details please see: https://doi.org/10.1016/j.nahs.2026.101730}}

\author[First,Second]{Ariana R. Méndez-Castillo} 
\author[Second]{Rodrigo Aldana-López} 
\author[First]{Antonio Ramírez-Treviño}
\author[Second]{Rosario Aragues}
\author[Third]{David Gómez-Gutiérrez}

\address[First]{Department of Electrical Engineering, Cinvestav-Guadalajara, Jalisco, 45019 México (e-mail: ariana.mendez@cinvestav.mx, antonio.ramirezt@cinvestav.mx)}
\address[Second]{Department of Computing and Systems Enginering - I3A, University of Zaragoza, Zaragoza 50018 España (e-mail: rodrigo.aldana.lopez@gmail.com, raragues@unizar.es ) }

\address[Third]{Tecnológico Nacional de México, Instituto Tecnológico José Mario Molina Pasquel y Henríquez, 
Cam. Arenero 1101, 45019 Zapopan, Jalisco, Mexico.\\
(e-mail: david.gomez.g@ieee.org)}

\begin{abstract}               
This work introduces a novel two-stage distributed framework to globally estimate constant parameters in a networked system, separating shared information from local estimation. The first stage uses dynamic average consensus to aggregate agents’ measurements into surrogates of centralized data. Using these surrogates, the second stage implements a local estimator to determine the parameters. By designing an appropriate consensus gain, the persistence of excitation of the regressor matrix is achieved, and thus, exponential convergence of a local Gradient Estimator (GE) is guaranteed. The framework facilitates its extension to switched network topologies, quantization, and the heterogeneous substitution of the GE with a Dynamic Regressor Extension and Mixing (DREM) estimator, which supports relaxed excitation requirements.

\end{abstract}


\end{frontmatter}

\section{Introduction}\label{sec:introduction}
The focus of this work is on estimating unknown constant parameters using linear regression data. This formulation naturally emerges in numerous applications, including adaptive control, cooperative robotics, power grids and sensor networks where reliable parameter identification is crucial for the stability and performance  of control strategies \citep{olfati2007consensus, Brouillon2024}. In particular, linear regression models appear in robot manipulator control relating torques, joint coordinates and the parameters of the robot \citep{gaz2019dynamic}.
 
Two architectural categories can be distinguished: the \emph{centralized} one, where agents transmit their information to a central unit that estimates the parameters, and the \emph{distributed} one, where agents estimate the global parameters locally through collective interaction. In the latter setting, each agent has access only to partial information.

In the centralized case, different algorithms have been proposed as estimators. For example, algorithms based on optimal filtering methods, such as the Kalman filter in discrete time \citep{kalman1960new, maybeck1982stochastic} and the Kalman–Bucy filter in continuous time \citep{golovan2002kalman} have been used in the form of recursive least-squares iterations. They model the unknown parameters as constant states and thus, provide unbiased estimates with minimum variance under additive Gaussian measurement noise. However, the convergence properties remain limited since the error covariance evolves according to a Riccati equation guaranteeing asymptotic convergence, but not exponential \citep{zhang2016convergence}. To solve this inconvenience, algorithms relying on gradient-based methods were introduced. For example, the Gradient Estimator (GE) \citep{ioannou1996robust, sastry2011adaptive} is widely used in both continuous- and discrete-time designs, and is usually formulated for noiseless measurements. For noisy measurements, the resulting estimation error can be explicitly analyzed, and a closed-form expression for the covariance can be derived \citep{MOSER2015149_GE}. Exponential convergence to the true parameters is guaranteed only when the persistence of excitation (PE) condition is maintained \citep{anderson2003exponential}. 
However, PE is a sufficient but not necessary condition for correct parameter estimation. To relax this requirement, the Dynamic Regressor Extension and Mixing (DREM) algorithm has been proposed \citep{aranovskiy2016drem,ortega2021parameter}. DREM reformulates the original regression into decoupled scalar equations, enabling parameter convergence under excitation conditions weaker than PE \citep{ortega2020DREM_PE}.

Unfortunately, centralized architectures have significant limitations, such as high bandwidth requirements, the computational power required by the central unit, and limited resilience to communication failures, among others \citep{kia2019tutorial}. These limitations are exacerbated when the number of agents increases. To face these challenges, distributed estimators have been introduced. Examples include distributed extensions of Kalman filter based parameter estimators \citep{lendek2007distributed, ryu2023consensus}, which only guarantee asymptotic convergence under the same conditions as their centralized counterparts. A widely studied alternative is the so called ``\emph{consensus+}\emph{innovations}'' framework \citep{kar2013consensus+, lorenz2025robust, papusha2014collaborative}. In this approach, each agent runs a tightly coupled gradient-based estimator with consensus enforcement (See Figure \ref{fig:network}-Top). Convergence relies on the notion of cooperative persistence of excitation (cPE) \citep{Chen2014PEcoope, yan2025distributed,Zheng2016SwitchGraphs,RO2021ESt_dis}, which ensures that the excitation condition is satisfied collectively at the network level. Variants based on least mean squares consensus adaptive filters have also been considered for different network topologies \citep{xie2018necessary}. A drawback of this framework is its tightly coupled structure, where consensus and parameter estimation are closely intertwined. In such approaches, convergence analyses depend on joint Lyapunov arguments that simultaneously treat both disagreement and estimation errors. Modifying or replacing the underlying estimator requires substantial changes to the entire algorithm and a new convergence analysis. Likewise, examining effects such as quantization in communication is not straightforward.

Lately, hierarchical architectures in multi-agent settings have been investigated for applications such as fault estimation \citep{liu2018hierarchical}, controller design \citep{Cheng2023Hierarchical}, and consensus algorithms \citep{chen2020hierarchical}, as they provide enhanced flexibility and heterogeneous local strategies. Such architectures also facilitate the integration of consensus and estimation layers, since they can be analyzed and implemented separately \citep{chen2020hierarchical}. In this context, hierarchical distributed parameter estimation schemes based on the DREM framework have been proposed \citep{RO2021ESt_dis, yan2025distributed}. However, these approaches either require local excitation conditions or fail to guarantee exponential convergence of the estimates. Consequently, no hierarchical method currently available in the literature achieves exponential convergence under the non-local cPE condition. 

Motivated by the previous discussion, the contributions of this work are as follows.
\begin{itemize}
    \item We introduce a hierarchical distributed parameter estimation framework that decouples consensus from local estimation. A Dynamic Average Consensus (DAC) block aggregates agents’ measurements into local surrogates of the centralized regression, followed by a local estimation block that recovers the unknown parameters. Compared to consensus+innovations methods, this separation allows independent analysis of consensus and estimation convergence.
\item The proposed framework is estimator agnostic. As a concrete instance, we present the first distributed DREM-based estimator operating under cPE, in contrast to existing hierarchical approaches which require local excitation.
\item The results are further extended to switching communication graphs and to scenarios with quantized information exchange.
\end{itemize}

{\bf Notation: }Let $a,b \in \mathbb{R}$. We write $a \ge b$ or $a \leq b$ to denote the usual order on \(\mathbb{R}\).  Let $\mathbb{R}^+$ denote the subspace of positive real numbers. The floor operator $\lfloor\bullet\rfloor$ denotes the greatest integer less than or equal to its argument. The identity matrix of dimension $n$ is denoted by $\mathbf{I}_n$. Let $\boldsymbol{1}=[1, \dots, 1]^\top\in\mathbb{R}^N$. The $i$-th canonical vector is denoted by $\mathbf{e}_i$ of appropriate dimension. The Kronecker product is denoted by $\otimes$. Let $\| \bullet \|$ denote the Euclidean norm for vector inputs, and induced Euclidean norm for matrices. Let $\mathbf{A}$ be a matrix. $|\mathbf{A}|$ represents the component-wise absolute value. The largest eigenvalue of $\mathbf{A}$ is denoted by $\lambda_{\max}(\mathbf{A})$.  The stacking operator for matrices is denoted as $\col({\mathbf{A}}_i)_{i=1}^N=[\mathbf{A}_1^\top \ \ \dots \ \ \mathbf{A}_N^\top]^\top$. The notation $\text{Tr}(\mathbf{A})$ represents the trace of the matrix $\mathbf{A}$. The operator $\text{vec}(\mathbf{A})$  denotes the vectorization of the matrix $\mathbf{A}$ obtained by stacking its columns into a single vector \citep[Page 60]{petersen2008matrix}. Denote with $\|\bullet\|_F=\|\text{vec}(\bullet)\|$ the Frobenius norm of a matrix. Let $\mathbf{A}\succ\mathbf{B}$ (resp. $\mathbf{A}\succeq\mathbf{B}$) be the denote when the matrix $\mathbf{A}-\mathbf{B}$ is positive definite (resp. positive semi-definite).

\section{Problem statement and preliminaries} \label{sec:guidelines}
In this work, we consider a team of $N$ monitoring agents in a networked system tasked to the estimation of a global vector of unknown constant parameters $\bm{\theta}\in\mathbb{R}^n$. Each agent is equipped with a local sensor that takes measurements of the form:
\begin{equation}\label{eq:original_system}
\mathbf{y}_i(t) = \mathbf{C}_i(t) \boldsymbol{\theta} 
\end{equation} 
where $\mathbf{y}_i(t) \in \mathbb{R}^{p_i}$ denotes the local sensor output of the $i$th agent and $\mathbf{C}_i(t) \in \mathbb{R}^{p_i \times n}$ is a time-varying regressor matrix known by agent $i\in\{1,\dots,N\}$. We focus on the nontrivial case in which the information available to each agent is insufficient to identify $\bm{\theta}$ on its own. In particular, for each agent $i$, the local regression matrix $\mathbf{C}_i(t)$ may be rank deficient over time, so that multiple parameter values are compatible with the local measurements, thus requiring cooperation among agents to recover $\bm{\theta}$.

In this framework, agents share information through a communication network modeled as a connected, undirected graph $\mathcal{G}=(\mathcal{V},\mathcal{E})$, with node set $\mathcal{V}=\{1,\dots,N\}$ and edge set $\mathcal{E}\subseteq \mathcal{V}\times\mathcal{V}$. An edge $(i,j)\!\in\!\mathcal{E}$ indicates that agents $i$ and $j$ share information between them. For network analysis purposes, let $\mathbf{A}=[a_{ij}] \in \mathbb{R}^{N \times N}$ denote the adjacency matrix of $\mathcal{G}$, with entries   $a_{ij}=a_{ji}=1$ when $(i,j)\in\mathcal{E}$ and $a_{ij}=a_{ji}=0$ otherwise. The degree matrix of $\mathcal{G}$ is denoted by the diagonal matrix $\mathbf{D}\in\mathbb{N}^{N\times N}$ with $d_{ii}=\sum_{j\neq i} a_{ij}$ and the Laplacian matrix is defined as $\mathbf{L}=\mathbf{D}-\mathbf{A}$. The set of neighbors of agent $i$ is denoted as $\mathcal{N}_i \subseteq \mathcal{V}$. The smallest nonzero eigenvalue $\lambda_{\mathcal{G}}$ of $\mathbf{L}$, called the algebraic connectivity, measures graph connectivity. 

The next subsection recalls the centralized parameter estimation version, where we highlight that the PE of a collective regressor matrix is sufficient to guarantee the exponential convergence of the gradient estimator to the true parameters. This is relevant since a similar result will be derived in our distributed proposal.

\subsection{Centralized estimation}
In the centralized estimation setting, there exists a central station that collects all measurements $\mathbf{y}_i(t)$ as $\mathbf{\mathbf{y}}(t) =\mathbf{C}(t)\boldsymbol{\theta}$, with 
\begin{equation}\label{Eq:MatCentralizado}
    \mathbf{y}(t)= \text{col}(\mathbf{y}_i(t))_{i=1}^N,\quad \mathbf{C}(t)=  \text{col}(\mathbf{C}_i(t))_{i=1}^N
\end{equation}
where $\mathbf{C}(t)\in \mathbb{R}^{p \times n}$ and $\mathbf{y}(t) \in \mathbb{R}^{p}$, with $p=\sum_{i=1}^Np_i$. 

Equipped with the pair $\mathbf{y}(t), \ \mathbf{C}(t)$, the GE \citep{ioannou1996robust} can be used to obtain an estimate $\bm{\hat{\theta}}_c$ of the true parameter $\bm{\theta}$ using an iteration in the form of the following differential equation:
\begin{equation}\label{eq:cent_estim}
    \dot{\bm{\hat{\theta}}}_c =\bm{\Gamma}_c  \mathbf{C}(t)^\top (\mathbf{y}(t) - \mathbf{C}(t) \bm{\hat{\theta}}_c ),
\end{equation}
where the gain matrix $\bm{\Gamma}_c \in \mathbb{R}^{n\times n}$, which satisfies $\boldsymbol{0} \prec \bm{\Gamma}_c$, determines the convergence rate of the estimation dynamics and can be tuned to accelerate convergence. It can be shown that \eqref{eq:cent_estim} corresponds to the gradient flow dynamics which minimizes the instantaneous estimation error: 
\begin{equation}
     J_c(\bm{\hat{\theta}}_c;t)=  \frac{1}{2}\|\mathbf{y}(t)-\mathbf{C} (t) \bm{\hat{\theta}}_c\|^2.
\end{equation}
To study when \eqref{eq:cent_estim} converges to the true parameter $\bm{\theta}$, the following standard definition is used:
\begin{definition}\label{def:PE}\citep[Adapted from Definition 4.3.1]{ioannou1996robust} The matrix function $\mathbf{F}:\mathbb{R}^{+}\to\mathbb{R}^{m\times n}$ is said to comply with the Persistency of Excitation (PE) condition if there exists $\alpha,T>0$ such that for all $t\geq 0$:
    \begin{equation}\label{eq:PE}
        \int_t^{t+T} \mathbf{F} (\tau)^\top\mathbf{F}(\tau) \text{d}\tau \succeq \ \alpha \mathbf{I}_n.
    \end{equation}
We denote a function $\mathbf{F}$ that satisfies the PE condition as $\mathbf{F}\in\text{PE}(\alpha, T)$ when the parameters $\alpha$ and $T$ are relevant and as $\mathbf{F}\in\text{PE}$ otherwise.
\end{definition}

The estimation error $\boldsymbol{\tilde{\theta}}_c = \boldsymbol{\hat{\theta}}_c -\boldsymbol{\theta}$ satisfies the following time-varying dynamical system:
\begin{equation}\label{eq:error_sge}
\begin{split}
    \dot{\tilde{\boldsymbol{\theta}}}_c = -\bm{\Gamma}_c \mathbf{C}(t)^\top \mathbf{C}(t) \tilde{\boldsymbol{\theta}}_c.
\end{split}
\end{equation}
The next result establishes the convergence of \eqref{eq:error_sge} to the origin when $\mathbf{C}\in \text{PE}$.
\begin{prop} \label{prop:error_PE} \citep[Adapted from Theorem 4.3.2]{ioannou1996robust}
Let \eqref{eq:error_sge} be the dynamics of the time-varying error, with $\bm{\Gamma}_c\succ \boldsymbol{0}$ and $\mathbf{C}\in\text{PE}$, then, the origin of \eqref{eq:error_sge} is globally uniformly exponentially stable.
\end{prop}
\begin{rem}
Proposition~\ref{prop:error_PE} shows that $\mathbf{C} \in $ PE is sufficient for exponential convergence of the centralized estimator, while, in general, $\mathbf{C}_i\notin\text{PE}$ for individual sensors. Thus $\mathbf{C}\in$ PE has become a standard assumption on centralized and cooperative settings. This condition is referred to as \emph{cooperative persistency of excitation (cPE)} \citep{Chen2014PEcoope,yan2025distributed,lorenz2025robust}. Requiring excitation at the local level, is significantly stronger and is only considered in specific settings, such as \citep{RO2021ESt_dis}. This motivates distributed mechanisms that exploit cooperation to recover cPE properties of the centralized model.
\end{rem}

\begin{figure}[t]
    \centering
    \includegraphics[width=0.7\linewidth]{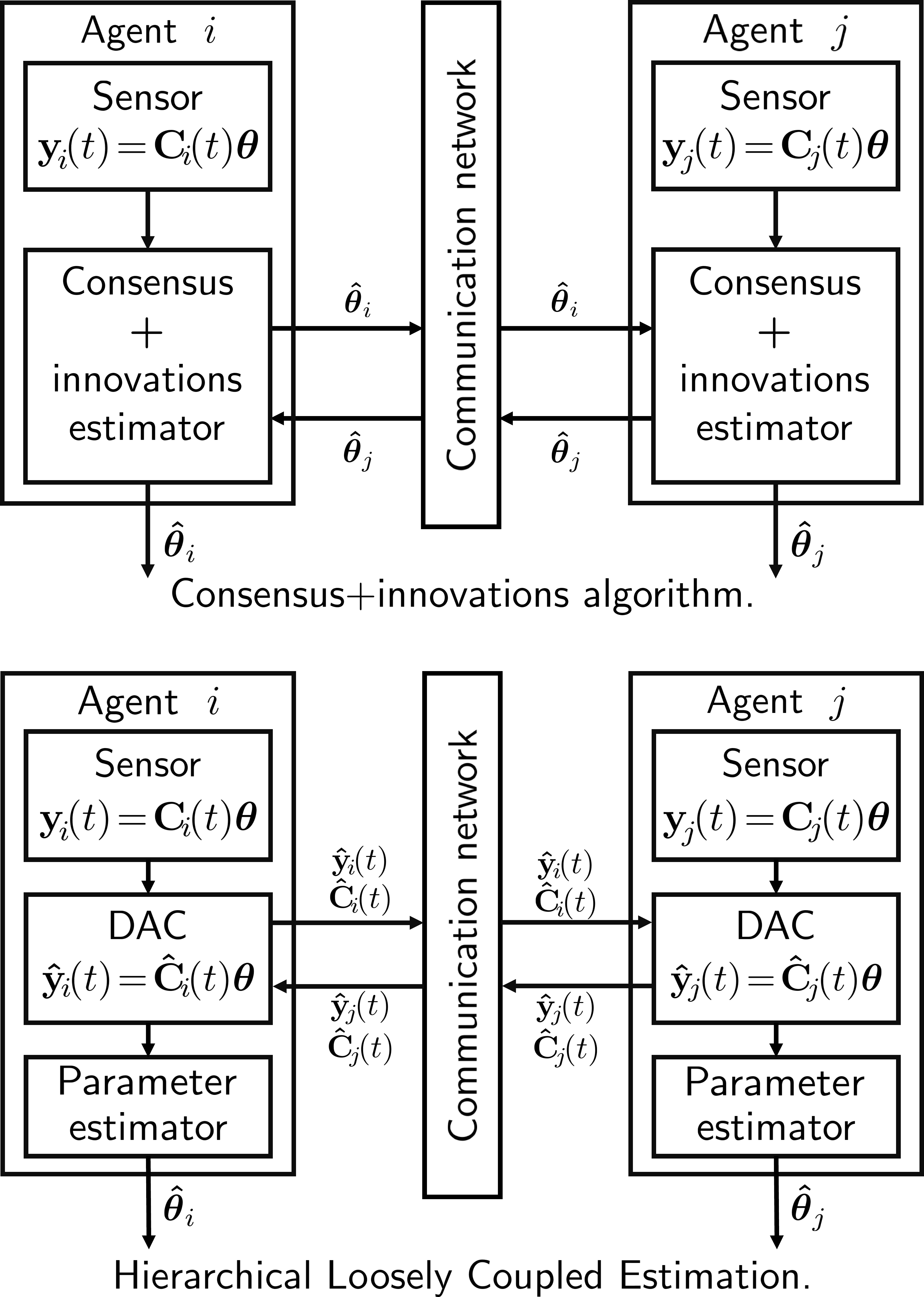}
    \caption{Distributed parameter estimator framework. (Top) Scheme of the consensus+innovations algorithm, implemented within the same block. (Bottom) Scheme of the proposed approach, where consensus and parameter estimation are decoupled.}
    \label{fig:network}
\end{figure}

\section{Hierarchical Loosely Coupled Estimation}\label{sec:Hierar_loos}
The proposed hierarchical framework equips each agent with a Dynamic Average Consensus (DAC) \citep{kia2019tutorial} block and a local parameter estimation block (see Figure \ref{fig:network}-Bottom). This structure decouples consensus dynamics from estimator convergence, avoiding the tight interdependence between them, observed in existing approaches. As a consequence, the framework is estimator-agnostic being capable of supporting different estimation schemes (e.g., GE, DREM), depending on the problem characteristics. In this section, we employ the GE scheme as a baseline, whereas the DREM estimator is introduced in Section \ref{sec:DREM}.

The first problem to tackle is the dimension mismatch among agents. For this purpose, the measurement in (\ref{eq:original_system}) is premultiplied as
\begin{equation}\label{eq:cons_i}
\mathbf{y}_i'(t) = \mathbf{C}_i(t)^\top\mathbf{y}_i(t) = \mathbf{C}_i'(t)\bm{\theta},
\end{equation}
where $\mathbf{y}_i'(t)\in\mathbb{R}^n$ and $\mathbf{C}_i'(t) = \mathbf{C}_i(t)^\top \mathbf{C}_i(t)$. Now $\mathbf{C}_i'(t) \in \mathbb{R}^{n\times n}$, and $\mathbf{y}_i'(t)$ have the same dimension in every agent. The  outputs $\mathbf{y}_i'(t)$ and the regressor matrices $\mathbf{C}_i'(t)$ of the agents are arranged in matrices  
\begin{equation}\label{eq:Matrix_centralized_consensus_regularized}
\mathbf{y}'(t)=\text{col}(\mathbf{y}'_i(t))_{i=1}^N \quad \mathbf{C}'(t)=\text{col}(\mathbf{C}'_i(t))_{i=1}^N\end{equation}

Equations $\mathbf{y}'(t) = \mathbf{C}'(t)\bm{\theta}$ and $\mathbf{y}(t) = \mathbf{C}(t)\bm{\theta}$ from  \eqref{Eq:MatCentralizado} have the same solutions, as stated in the following proposition.

\begin{prop}\label{prop:prime_equivalence}
    Let $\mathbf{y}(t),\mathbf{C}(t)$ and $\mathbf{y}'(t),\mathbf{C}'(t)$ be defined as in \eqref{Eq:MatCentralizado} and \eqref{eq:Matrix_centralized_consensus_regularized}, respectively. 
    Then, any $\bm{\theta} \in \mathbb{R}^n$ is a solution $\mathbf{y}'(t) = \mathbf{C}'(t) \bm{\theta}$ if and only if $\bm{\theta}$ is a solution to $\mathbf{y}(t)=\mathbf{C}(t)\bm{\theta}$.
\end{prop}
All proofs are placed in Appendix \ref{sec:proofs}.

Next, we specify the consensus dynamics for the surrogate regressions $\mathbf{C}'_i(t)$ together with the distributed estimation of the parameter vector $\bm{\theta}$. For this purpose, we employ \eqref{eq:cons_i} in place of \eqref{eq:original_system}, since, as shown in Proposition~\ref{prop:prime_equivalence}\ , both formulations share the same solution set for $\bm{\theta}$. Based on this representation, we introduce the following algorithm, which integrates the consensus and distributed estimation layers.
 \begin{subequations}\label{alg:distributed_estimator}
\begin{align}
    &\textbf{Dynamic Average Consensus:} \nonumber\\
    &\dot{\mathbf{X}}_i(t) = k \sum_{j\in\mathcal{N}_i}\big(\hat{\mathbf{C}}_i(t)-\hat{\mathbf{C}}_j(t)\big)
    \label{eq:Consensus_P}\\
    &\dot{\mathbf{x}}_i(t) = k \sum_{j\in\mathcal{N}_i}\big(\hat{\mathbf{y}}_i(t)-\hat{\mathbf{y}}_j(t)\big)\label{eq:consensus_x}\\
    &\textbf{Consensus output:} \nonumber\\
    &\hat{\mathbf{C}}_i(t) = \mathbf{C}_i'(t) - \mathbf{X}_i(t), \label{eq:Consensus_output}\\
    &\hat{\mathbf{y}}_i(t) = \mathbf{y}_i'(t) - \mathbf{x}_i(t), \label{eq:consensus_output_y} \\
    &\textbf{Local parameter estimation:} \nonumber\\
    &\dot{\boldsymbol{\hat{\theta}}}_i(t) 
      = \bm{\Gamma}_i \hat{\mathbf{C}}_i(t)^\top\big(\hat{\mathbf{y}}_i(t) - \hat{\mathbf{C}}_i(t)\boldsymbol{\hat{\theta}}_i(t)\big) \label{eq:alg_SGE}
\end{align}
\end{subequations}
where $k>0$, $\boldsymbol{0} \prec\bm{\Gamma}_i\in\mathbb{R}^{n\times n}$ are design parameters. The gains $k$ and $\bm{\Gamma}_i$ determine the convergence rates of the consensus and estimation dynamics respectively, and can be tuned to accelerate convergence. As we show below, an appropriate selection of $k$ also ensures the PE property of the estimator regressor matrix. In addition, we restrict the algorithm to the standard DAC initialization $
    \sum_{i=1}^N\mathbf{X}_i(0) = \bm{0},\quad \sum_{i=1}^N\mathbf{x}_i(0) = \bm{0}$
with trivial setting $\mathbf{X}_i(0)=\bm{0}$ and $\mathbf{x}_i(0)=\bm{0}$. 
\begin{rem}
\label{rem:intuition}
The DAC block in \eqref{eq:Consensus_P}-\eqref{eq:consensus_output_y} allows agents to cooperatively reconstruct global quantities of interest that would otherwise require centralized aggregation, producing the surrogates $\hat{\mathbf{C}}_i(t)$ and $\hat{\mathbf{y}}_i(t)$. 
The outputs of the DAC algorithm in \eqref{eq:Consensus_output}--\eqref{eq:consensus_output_y} at each agent track the time-varying average of the reference signals. Ideally, the DAC dynamics aim to enforce $\hat{\mathbf{C}}_i(t)=\overline{\mathbf{C}}(t)$ and $\hat{\mathbf{y}}_i(t)=\overline{\mathbf{y}}(t)$, where
    \begin{equation}\label{eq:DAC_limit}
    \overline{\mathbf{C}}(t) = \frac{1}{N}\sum_{j=1}^N \mathbf{C}_j'(t), 
    \qquad 
    \overline{\mathbf{y}}(t) = \frac{1}{N}\sum_{j=1}^N \mathbf{y}_j'(t),
    \end{equation}
    equally  \begin{equation}\label{Eq:regresor}
        \hat{\mathbf{y}}_i(t) = \hat{\mathbf{C}}_i(t)\bm{\theta}
    \end{equation} 
    which motivates the structure of the second block, namely the local parameter estimator in \eqref{eq:alg_SGE}.
\end{rem}
The following assumption is required to show the convergence of \eqref{alg:distributed_estimator}, with similar analogs in the literature \citep{lorenz2025robust}:
\begin{assumption}\label{as:bounds}
There exists known $\beta, \gamma>0$ with $\gamma<\infty$, such that:
 $$\text{sup}_{\tau\in [t,\infty)}\!\|(\mathbf{H}\otimes \mathbf{I}_n)\dot{\mathbf{{C}}}'(\tau)\|\leq \gamma,\quad\|\mathbf{\overline{C}}(t)\|\leq \beta $$
 for all $t\geq 0$, where $
\mathbf{H} = \mathbf{I}_N - \bm{1}\bm{1}^\top/N.$
\end{assumption}

\begin{rem}
\label{rem:as:bound}
Assumption~\ref{as:bounds} is a sufficient regularity condition introduced to enable the analytical development and convergence guarantees of the proposed estimators. It is satisfied by a broad class of practically relevant regressors, including those generated by bounded linear systems and structured deterministic signals, as commonly assumed in adaptive and distributed estimation settings \citep{ioannou1996robust,papusha2014collaborative,RO2021ESt_dis,lorenz2025robust}. Consequently, the theoretical guarantees derived in this work may not apply to all classes of regressors, and extending the analysis beyond the considered class is left for future work.
\end{rem}

The following theorem shows that it is not necessary for the consensus outputs, $\mathbf{\hat{C}}_i(t)$ and $\mathbf{\hat{y}}_i(t)$, to be exactly equal to \eqref{eq:DAC_limit}. Instead, suitable design conditions ensure that the output $\hat{\bm{\theta}}_i(t)$ of \eqref{alg:distributed_estimator} converges to the true parameter $\bm{\theta}$.

\begin{thm}\label{the:gain_PE}
    Suppose Assumption \ref{as:bounds} holds and let $\mathbf{C}(t)$, defined in \eqref{Eq:MatCentralizado}, be $\text{PE}(\alpha,T)$ for some known $\alpha,T>0$. If in Algorithm \eqref{alg:distributed_estimator} the consensus gain $k$ satisfies
    \begin{equation}\label{eq:gain_k}
        k \; > \; \frac{2 n N^2\beta \gamma T^2}{\lambda_\mathcal{G} \alpha^2},
    \end{equation}
    then the origin of the estimation error $\hat{\bm{\theta}}_i(t)-\bm{\theta}$ is globally uniformly exponentially stable for all $i\in\{1,\dots,N\}$.
\end{thm}
\begin{rem}\label{rem:lambda}
    The lower bound for feasible values of $k$ in \eqref{eq:gain_k} can be conservatively estimated by overestimating $N$ and underestimating $\lambda_{\mathcal G}$. There exist distributed estimation approaches \citep{Montijano2011N_estimation,montijano2017fast} in which each agent independently estimates a value of $\lambda_G$.
     When the maximum number of nodes $N$ is known (for example, from network capacity specified by the communication protocol), the smallest possible $\lambda_{\mathcal G}$ for a connected graph can be used, which corresponds to that of a linear graph, i.e. $\lambda_{\mathcal G}=2(1-\cos(\pi/N))$ \citep[Table 3.2]{de2007old_Connect}. 
\end{rem}

Since inter-agent communication typically occurs over digital channels with finite data rates, we next account for bounded quantization effects in the consensus layer. In particular, we replace the DAC block in \eqref{eq:Consensus_P},\eqref{eq:consensus_x} of Algorithm~\eqref{alg:distributed_estimator}
by the following quantized version:
\begin{subequations}\label{alg:distributed_estimator:quantized}
\begin{align}
    &\textbf{Quantized Dynamic Average Consensus:} \nonumber\\
    &\dot{\mathbf{X}}_i(t)
    =k\!\!\sum_{j\in\mathcal{N}_{i}}\!
    \big(\mathcal{Q}(\hat{\mathbf{C}}_i(t))-\mathcal{Q}(\hat{\mathbf{C}}_j(t))\big),
    \label{eq:Consensus_P:quant}\\
    &\dot{\mathbf{x}}_i(t)
    =k\!\!\sum_{j\in\mathcal{N}_{i}}\!
    \big(\mathcal{Q}(\hat{\mathbf{y}}_i(t))-\mathcal{Q}(\hat{\mathbf{y}}_j(t))\big),
    \label{eq:Consensus_x:quant}
\end{align}
\end{subequations}
while keeping the consensus outputs
\eqref{eq:Consensus_output}, \eqref{eq:consensus_output_y}
and the local parameter estimator \eqref{eq:alg_SGE} unchanged. The quantization operator $\mathcal{Q}(\bullet)$ is defined componentwise as
\[
\mathcal{Q}(s)=\varepsilon\big\lfloor s/\varepsilon\big\rfloor,
\]
and is applied element-wise when acting on vectors or matrices, where $\varepsilon> 0$ denotes the size of the quantization step.
\begin{cor}\label{cor:quantization}
Let $\varepsilon> 0$ be the size of the quantization step. Assume that $\mathbf{C}(t)$ in ~\eqref{Eq:MatCentralizado} is $\text{PE}(\alpha,T)$ for some known $\alpha,T>0$, 
and Assumption~\ref{as:bounds} holds. Consider Algorithm~\eqref{alg:distributed_estimator} with the quantized consensus dynamics \eqref{eq:Consensus_P:quant}, \eqref{eq:Consensus_x:quant} and the local estimator \eqref{eq:alg_SGE}. If $k$ and $\varepsilon$ satisfy
\[
\frac{\alpha^2}{T N^2}
\;>\;
2\beta T\left(\frac{n\gamma}{k\lambda_{\mathcal G}}
+
\frac{n^2 \sqrt{N} \lambda_{\max}(\mathbf L)}{\lambda_{\mathcal G}}\,\varepsilon\right),
\]
then there exists $T_q>0$ such that the estimation error complies $\|\hat{\bm{\theta}}_i(t)-\bm{\theta}\|\leq B(\varepsilon)$ for all $t\geq T_q$, where $B(\bullet)$ is continuous, strictly increasing and comply $B(0)=0$. 
\end{cor}

Corollary \ref{cor:quantization} can be extended to a switching network topology $\mathcal{G}_{\sigma(t)}$ with a switching signal $\sigma(t)\in\{1,\dots,q\}$, selecting among a family of connected undirected graphs $\{\mathcal{G}_1,\dots,\mathcal{G}_q\}$, by replacing Dynamic Average Consensus block in \eqref{eq:Consensus_P:quant} and \eqref{eq:Consensus_x:quant} with: 
\begin{subequations}\label{alg:distributed_estimator:switched}
\begin{align}
    &\textbf{Switched  Dynamic Average Consensus:} \nonumber\\
    &\dot{\mathbf{X}}_i(t)
    =k\!\!\sum_{j\in\mathcal{N}_{i,\sigma(t)}}\!
    \big(\mathcal{Q}(\hat{\mathbf{C}}_i(t))-\mathcal{Q}(\hat{\mathbf{C}}_j(t))\big),
    \label{eq:Consensus_P:switched}\\
    &\dot{\mathbf{x}}_i(t)
    =k\!\!\sum_{j\in\mathcal{N}_{i,\sigma(t)}}\!
    \big(\mathcal{Q}(\hat{\mathbf{y}}_i(t))-\mathcal{Q}(\hat{\mathbf{y}}_j(t))\big),
    \label{eq:Consensus_x:switched}
\end{align}
\end{subequations}
where $\mathcal{N}_{i,\sigma(t)}$ is the neighbor set of agent $i$.
\begin{cor}\label{cor:switched}
    Under the Assumptions of Corollary \ref{cor:quantization}, let  $\mathcal{G}_{\sigma(t)}$ be a switching network topology with $\mathcal{G}_\ell$, $\ell=1,\ldots,q$ a connected undirected graph and the switching signal $\sigma(t)$ satisfying a minimum dwell-time. Moreover, let $k$ and $\varepsilon$ satisfy
\[
\frac{\alpha^2}{T N^2}
\;>\;
2\beta T\left(\frac{n\gamma}{k\lambda_{\mathcal{G},m}}
+
\frac{n^2 \sqrt{N} \lambda_{\mathcal{G},M}}{\lambda_{\mathcal{G},m}}\,\varepsilon\right),
\]
where $\lambda_{\mathcal{G},M}=\max_{\ell \in \{1,\dots q\}} \lambda_{\max}(\mathbf L_\ell)$ and $\lambda_{\mathcal{G},m}=\min_{\ell \in \{1,\dots q\}}\lambda_\ell$ and $\lambda_\ell,\mathbf{L}_\ell$ are the algebraic connectivity and Laplacian of $\mathcal{G}_\ell$ respectively. Then, under
algorithm
\eqref{eq:Consensus_P:switched}, \eqref{eq:Consensus_x:switched} together with \eqref{eq:Consensus_output}, \eqref{eq:consensus_output_y} and \eqref{eq:alg_SGE}, there exists $T_q>0$ such that the estimation error complies $\|\hat{\bm{\theta}}_i(t)-\bm{\theta}\|\leq B(\varepsilon)$ for all $t\geq T_q$, where $B(\bullet)$ is continuous strictly increasing with $B(0)=0$. 
\end{cor}

As shown in Appendix \ref{sec:convergence}, the proof of Corollary \ref{cor:quantization} is based on a graph-independent Lyapunov function. Thus, the proof of Corollary~\ref{cor:switched} follows from a common Lyapunov function argument with the dwell-time condition ensuring the well-posedness of the resulting dynamics \citep{liberzon2003switching}.

\section{Distributed DREM}\label{sec:DREM}
The DREM approach requires a weaker PE condition to guarantee the convergence of the estimated parameters. Therefore, it can be used as an alternative to the GE and may remain effective even if the GE fails.

The DREM \citep{aranovskiy2016drem} proposes applying linear, $\mathcal{L}_\infty$-stable operators $\mathcal{H}_i$ to the linear regression equation. In particular, for example, exponentially stable linear time-invariant filters of the form
\begin{equation}
    \mathcal{H}_j(\bullet)=\frac{\alpha_j}{p + \beta_j}(\bullet)
\end{equation}
where $p=\tfrac{d}{dt}$, and $\alpha_j \neq 0, \beta_j >0$ can be used. These filters introduce additional information, generally seeking to have a square matrix from which a scalar equation for each parameter is obtained after multiplication with its adjugate.

In our case, we propose introducing linear, $\mathcal{L}_\infty$-stable operators on $\mathbf{\hat{C}}_i(t)$ and $\mathbf{\hat{y}}_i(t)$, presented in equations \eqref{eq:Consensus_output}--\eqref{eq:consensus_output_y}, stacking the result in matrices $\mathbf{C}^f_i(t)$ and $\mathbf{y}_i^f(t)$, respectively, until $\det(\mathbf{C}_i^f(t)^\top \mathbf{C}_i^f(t)) \notin \mathcal{L}_2$. Thus, the GE block in \eqref{eq:alg_SGE} is replaced by a DREM estimator as follows: \begin{subequations}\label{alg:DREM_Filters}
    \begin{align}
     &\textbf{Local parameter estimation:} \nonumber\\
        & \mathbf{C}^f_i(t)=[\mathbf{\hat{C}}_i(t)^\top \ \ \mathcal{H}_1[\mathbf{\hat{C}}_i(t)]^\top  \ \ \dots \ \ \mathcal{H}_{r}[\mathbf{\hat{C}}_i(t)]^\top ]^\top \label{eq:alg_Drem_C}\\
        & \mathbf{y}_i^f (t) = [\mathbf{\hat{y}}_i(t)^\top \ \ \mathcal{H}_1[\mathbf{\hat{y}}_i(t)]^\top  \ \ \dots \ \ \mathcal{H}_{r}[\mathbf{\hat{y}}_i(t)]^\top ]^\top \label{eq:alg_drem_y}\\
        & \phi_i(t) = \text{det} (\mathbf{C}^f_i(t)^\top \mathbf{C}^f_i(t))\label{eq:alg_phi} \\
        & \mathbf{Y}_i(t)= \text{adj}(\mathbf{C}^f_i(t)^\top \mathbf{C}^f_i(t)) (\mathbf{C}^f_i (t)^\top \mathbf{y}_i^f(t)) \label{eq:alg_DREM_Yg}\\
        &\dot{\boldsymbol{\hat{\theta}}}_i(t) = \bm{\Gamma}_i \phi_i(t)\big(\mathbf{Y}_i(t) - \phi_i(t)\boldsymbol{\hat{\theta}}_i(t)\big). \label{eq:alg_DREM_th}
    \end{align}
\end{subequations}
where $\boldsymbol{\Gamma}_i \succ \boldsymbol{0}$ is a diagonal matrix and a design parameter and $\mathcal{H}_j: \mathcal{L}_\infty \to \mathcal{L}_\infty$ is a linear, $\mathcal{L}_\infty$-stable operator with $j \in \{1, \dots, r\}$.

\begin{cor}\label{cor:DREM_L2}
Let \eqref{eq:Consensus_output}, \eqref{eq:consensus_output_y} be consensus outputs and let $\mathbf{C}_i^f(t)$ be defined in \eqref{eq:alg_Drem_C}. If the number $r$ of operators $\mathcal{H}_j$ makes $\phi_i (t)\notin \mathcal{L}_2$, then $\hat{\boldsymbol{\theta}}_i(t)$ in \eqref{eq:alg_DREM_th} converges exponentially to $\boldsymbol{\theta}$.
\end{cor}

Algorithm \eqref{alg:DREM_Filters} is based on introducing the necessary number of operators $\mathcal{H}_j$ that allow us to expand the information available to each agent. However, the consensus output could fulfill condition \eqref{eq:condition_DREM} on its own, so it would not be necessary to introduce operators, i.e., $r=0$.
 
 Furthermore, when the matrix $\mathbf{\hat{C}}_i(t)$ resulting from the consensus in \eqref{eq:Consensus_output} is PE or $\phi_i=\det(\mathbf{\hat{C}}_i(t))$ satisfies condition \eqref{eq:condition_DREM}, a simpler version of the local DREM-based parameter estimator is given by
\begin{subequations}
    \begin{align}\label{alg:DREM_NotFilter}
        &\textbf{Local parameter estimation:} \nonumber\\
        & \phi_i(t) = \det\!\left(\mathbf{\hat{C}}_i(t)\right), \\
        & \mathbf{Y}_i(t)= \operatorname{adj}\!\big(\mathbf{\hat{C}}_i(t)\big)\,\mathbf{\hat{y}}_i(t), \\
        & \dot{\boldsymbol{\hat{\theta}}}_i(t) 
          = \bm{\Gamma}_i\,\phi_i(t)\left(\mathbf{Y}_i(t) - \phi_i(t)\boldsymbol{\hat{\theta}}_i(t)\right). 
          \label{eq:alg_DREM_cor}
    \end{align}
\end{subequations}
\begin{cor}\label{cor:DREM}
Let $\mathbf{\hat{C}}_i(t)$ and $\mathbf{\hat{y}}_i(t)$ be the consensus outputs defined in~\eqref{eq:Consensus_output}--\eqref{eq:consensus_output_y}. If $\mathbf{\hat{C}}_i(t)$ is PE or $\det(\mathbf{\hat{C}}_i(t)) \notin \mathcal{L}_2$, then $\hat{\boldsymbol{\theta}}_i$ in \eqref{eq:alg_DREM_cor} converges exponentially to $\boldsymbol{\theta}$. 
\end{cor}

\section{Discussion}\label{sec:discussion}
Based on the obtained results,
we emphasize that the proposed architecture (see bottom Figure \ref{fig:network}) for distributed parameter estimation is both flexible and versatile, as it enables the decoupling of the consensus analysis from the parameter estimation process. The architecture is composed of two blocks: a consensus module and a parameter estimation module. In the consensus stage, one works with surrogate versions of the regressor matrices instead of the original ones. A key advantage of this approach is that the surrogate signals $\hat{\mathbf{C}}_i(t)$ exhibit local PE  even in cases where the original regressor matrices do not have this property. This is guaranteed by selecting the gain according to a computable bound (see Theorem \ref{the:gain_PE}). This allows us to formulate a new regressor equation that relies on the surrogates in place of the original regressor matrices in \eqref{Eq:regresor} (see Lemma \ref{lem:regression_invariance}).  Because the surrogates possess the property of persistent excitation, the parameter estimation process is ensured to converge exponentially to the true parameter values. That is to say, the consensus module does not act directly on the parameter estimates, but rather on the regression data. By exploiting the separation between analyses, the architecture becomes modular, allowing different algorithms to be applied at each stage. This flexibility enables it to capture various phenomena found in real-world problems. For instance, during the consensus phase, it is possible to incorporate effects like data quantization/rounding and switching communication graphs, while still demonstrating that the surrogates retain the persistent excitation property. Furthermore, in the estimation phase, we demonstrated that both gradient-descent and DREM-based estimation methods can be employed, and in each case the true parameter values are successfully recovered.

\section{Numerical Examples}\label{sec:example}

We consider $N = 10$ agents connected by a collection of undirected graphs $\G_\sigma$ given in Figures \ref{fig:graph1}-\ref{fig:graph4}. The switched network evolves according to a switching signal $\sigma(t)\in\{1,2,3,4\}$. For each agent $i$, the measurements are given by $y_i(t) = \mathbf{C}_i(t)\boldsymbol{\theta}$ with $n=3$, where $\mathbf{C}_i(t)\in \mathbb{R}^{1\times n}$ consists of terms of the form $A + B \sin(\omega t) + D \cos(\omega t)$. The parameters $\omega, A, B,$ and $D$ are sampled from uniform distributions, with $A, B, D$ sampled over $[0,20]$ and $\omega$ sampled over $[0,3]$. Based on these values, simulations yield $\gamma=8.3966$, $\beta=20.769$, $\alpha=51.326$, $T=0.16$, and $\min_{\sigma\in\{1,2,3,4\}}\lambda_\sigma=0.367$. Using these parameters, the consensus gain must satisfy $k> 2.778$. In the simulations, we select $k=2.806$.
\begin{figure}
    \begin{subfigure}{0.23\textwidth}
    \centering
        \includegraphics[width=\linewidth]{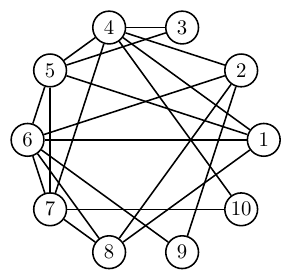}
        \caption{$\sigma=1$}
        \label{fig:graph1}
    \end{subfigure}
    \hfill
    \begin{subfigure}{0.23\textwidth}
    \centering
        \includegraphics[width=\linewidth]{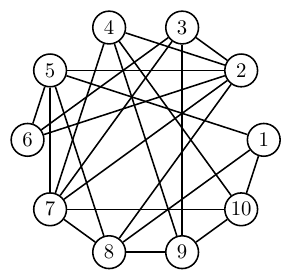}
        \caption{$\sigma=2$}
        \label{fig:graph2}
    \end{subfigure}
    \medskip
    \begin{subfigure}{0.23\textwidth}
    \centering
        \includegraphics[width=\linewidth]{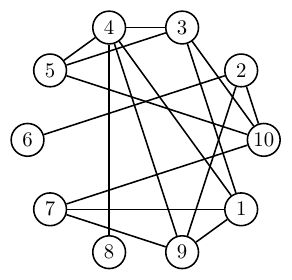}
        \caption{$\sigma=3$}
        \label{fig:graph3}
    \end{subfigure}
    \hfill
    \begin{subfigure}{0.23\textwidth}
    \centering
        \includegraphics[width=\linewidth]{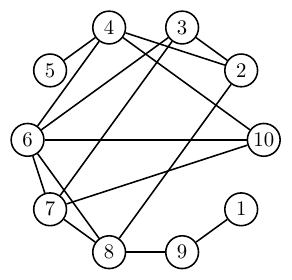}
        \caption{$\sigma=4$}
        \label{fig:graph4}
    \end{subfigure}
    \caption{Collection of communication topologies $\G_\sigma$ used for the switched dynamic network in the numerical example.}
    \label{fig:graphs}
\end{figure}

The examples in this section illustrate that both estimation approaches, distributed GE and DREM, achieve good performance. In all figures, blue corresponds to $\theta_1$, red to $\theta_2$, and green to $\theta_3$. In the parameter estimation graphs, the gradient-based approach is depicted in green, while DREM is depicted in blue. Let $\tilde{\theta}_{i,\mu}$ represent the estimation error of agent $i$ with respect to $\theta_\mu, \mu=1,\dots,n$. In the estimation error plots, we group all errors related to $\theta_\mu$ into the vector $\mathbf{e}_\mu = [\,\tilde{\theta}_{1,\mu}, \dots, \ \tilde{\theta}_{N,\mu}]$.

Figure \ref{fig:example} shows the simulation results with a quantization step size of zero.  The upper graph shows the estimate for each agent’s parameter estimate, while the dashed lines indicate the actual parameter values to be identified. The GE is shown in green and the DREM estimator in blue. In both cases, the estimates converge exponentially to the true values. The upper-middle graph shows the consensus output, where the dotted line represents the average value $\overline{\mathbf{y}}(t)$, in \eqref{eq:DAC_limit}, and the solid lines represent each consensus output $\hat{\mathbf{y}}_i(t)$ in \eqref{eq:consensus_output_y}. The lower-middle plot presents the estimation error $\|\mathbf{e}_i\|$ for $i \in \{1, 2, 3\}$. The dashed line corresponds to the GE error, while the solid line represents the DREM estimator error. In both cases, the errors converge exponentially to zero. The graphs switch according to the signal $\sigma(t)$, shown in the lower panel.

\begin{figure}
    \centering
    \includegraphics[width = \linewidth]{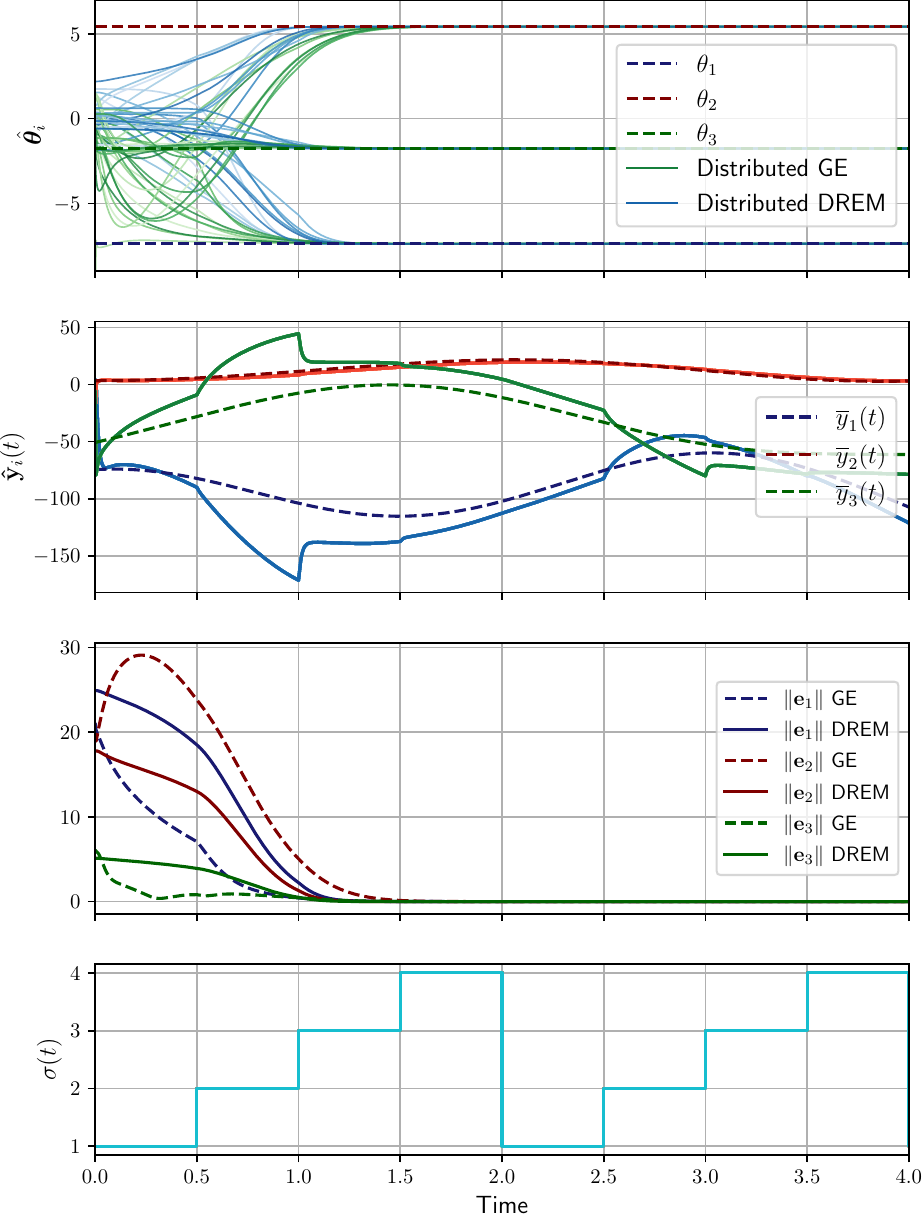}
    \caption{Simulation results. Colors are related to the parameters $\theta_1$ (blue), $\theta_2$ (red), and $\theta_3$ (green). (Top) Parameters estimation $\hat{\theta}_{i,\mu},$ for each agent for GE (green) and DREM (blue); dashed lines show $\bm{\theta}$. (Top-middle) Dashed line represents the average value ($\overline{\mathbf{y}}(t)$), while solid lines represent the individual consensus outputs $\hat{\mathbf{y}}_i(t)$. (Low-middle) Magnitude of the estimation error $\mathbf{e}_i$, dashed for GE and solid for DREM. (Bottom) The switching signal $\sigma(t)$.}
    \label{fig:example}
\end{figure}

In the next experiment, we use a quantization step of $\varepsilon = 0.036$, in accordance with the assumptions of Corollary \ref{cor:switched}, and keep the consensus gain value unchanged. The result is shown in Figure \ref{fig:quantization}. 
The upper graph displays the estimated parameters, revealing a small offset. The lower graph shows that, due to the quantization error, the estimation error does not exactly converge to zero, in agreement with Corollary \ref{cor:switched}.

To further illustrate the capabilities of the proposed approach, we examine the case in which noise is added to sensor measurements. Specifically, we take $\mathbf{y}_i(t) = \mathbf{C}_i(t)\bm{\theta} + \bm{\eta}_i (t)$, where we sample $\bm{\eta}_i(t)$ as zero-mean Gaussian noise with a standard deviation equal to $0.2$ for each simulation step. As shown in Figure \ref{fig:Noise}, the parameter estimates are affected by noise. Nevertheless, the estimation error still converges exponentially to a neighborhood of the origin. Because the noise scales with the consensus gain, increasing this gain results in a poorer estimate. This outcome suggests that the parameter can be recovered with reasonable accuracy in presence of noise.

Finally, we consider the scenario in which information is lost during communication between agents. The simulation results are presented in Figure \ref{fig:loss_paq}. In the upper graph, the estimation error decreases at an exponential rate. The lower graph illustrates the information flow by each agent, arranged from bottom to top. A high level represent periods when the agent is sending information, whereas low level indicate that the agent is not transmitting.
\begin{figure}
    \centering
    \includegraphics[width=1\linewidth]{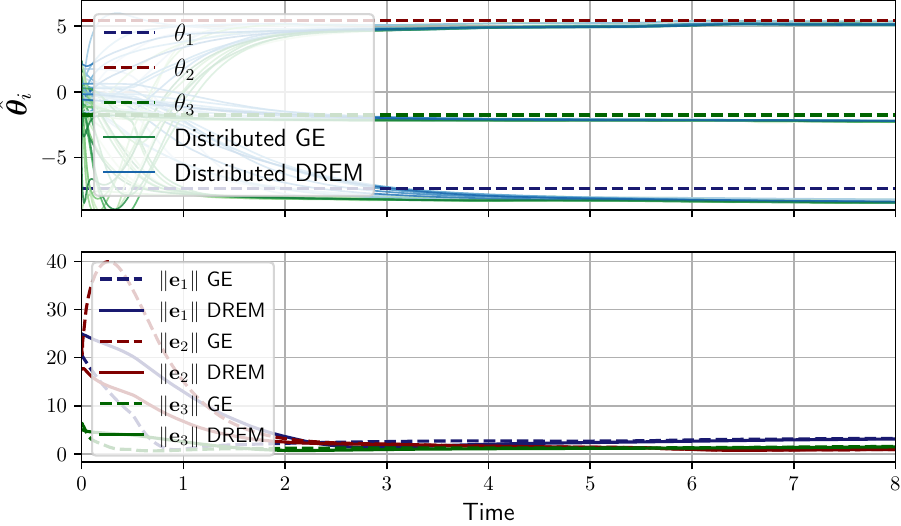}
    \caption{Simulation results using quantization. Colors are related to the parameters $\theta_1$ (blue), $\theta_2$ (red) $\theta_3$ (green). (Top) Parameter estimation for each agent for GE (green) and DREM (blue); dashed lines show $\bm{\theta}$. (Bottom) Magnitude of estimation error $\mathbf{e}_i$, dashed lines for GE and solid for DREM.}
    \label{fig:quantization}
\end{figure}

\begin{figure}
    \centering
    \includegraphics[width=\linewidth]{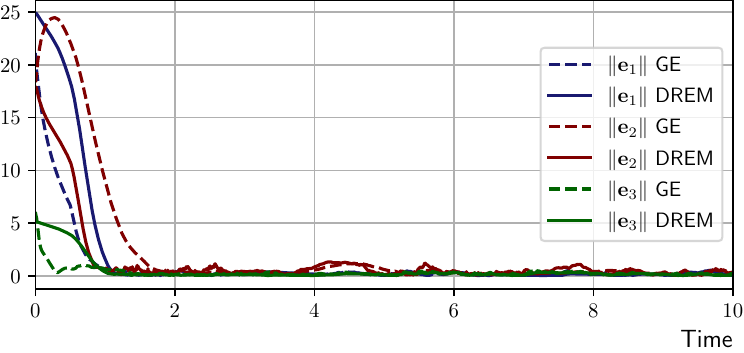}
    \caption{Simulation results with additive Gaussian sensor noise modeled as $\mathbf{y}_i(t)=\mathbf{C}_i(t)\bm{\theta}+\bm{\eta}_i(t)$, where $\bm{\eta}_i(t)$ is Gaussian noise with mean $0$ and standard deviation equal to $0.2$. Estimation error magnitudes $\mathbf{e}_1$ (blue), $\mathbf{e}_2$ (red), and $\mathbf{e}_3$ (green) are shown; dotted lines indicate GE and solid lines DREM.}
    \label{fig:Noise}
\end{figure}

\begin{figure}
    \centering
    \includegraphics[width=\linewidth]{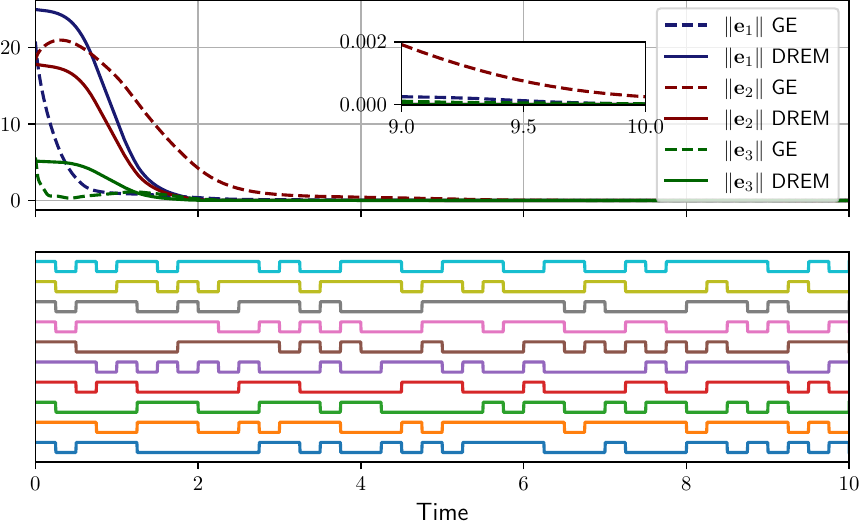}
    \caption{Simulation results with communication loss. (Top) Magnitude of estimation errors $\mathbf{e}_1(t)$ (blue), $\mathbf{e}_2$ (red), and $\mathbf{e}_3$ (green), where dotted lines correspond to GE and solid lines to DREM. (Bottom) Information transmitted by each agent. High level indicates that information is being sent, while the low level indicates no transmission.}
    \label{fig:loss_paq}
\end{figure}

\section{Conclusion}
This work presented a hierarchical distributed estimation framework that decouples the DAC algorithm from the parameter estimation. As a result, by calculating an appropriate consensus a regression matrix can be obtained that satisfies the PE condition, even when the communication topology switches over time. Under this condition, exponential convergence is ensured for two distinct parameter estimation methods: GE and DREM, enabling flexibility in the choice of estimator while preserving the advantages of each. The framework further allows for quantized communication between agents, yielding an estimation error that remains ultimately bounded, with its magnitude dictated by the chosen quantization step. All results are validated through a Lyapunov-function-based analysis, which guarantees global and exponential stability.
The results are confirmed by the simulations. For both GE and DREM, the estimation error exhibits exponential convergence. This behavior holds across all analyzed scenarios and aligns with the theoretical predictions. Furthermore, even in the presence of communication losses and sensor measurement noise, the estimators exhibit satisfactory performance, illustrating the practical robustness of the proposed framework. As future work, we aim to investigate adaptive gain tuning laws to overcome the limitation of knowing the lower bound for admissible consensus gain in advance.
Moreover, the impact of communication losses and measurement noise is yet to be analyzed formally. In addition, the proposed framework may be extended to nonlinear measurement models and event-triggered protocols.

\appendix\section{Proofs}
\label{sec:proofs}
\subsection{Proof of Proposition \ref{prop:prime_equivalence}}\label{appen:prop2}

Set $\mathcal{W}=\{\bm{\theta}\in \mathbb{R}^n : \mathbf{y}(t)=\mathbf{C}(t)\bm{\theta}\}$ and $\mathcal{W}'=\{\bm{\theta}\in \mathbb{R}^n : \mathbf{y}'(t)=\mathbf{C}'(t)\bm{\theta}\}$. Take any $\bm{\theta}\in \mathcal{W}'$.  Then $\mathbf{y}'(t) = \mathbf{C}'(t)\bm{\theta} = \mathbf{C}(t)^\top\mathbf{C}(t)\bm{\theta}$. Since $\mathbf{y}'(t)=\mathbf{C}(t)^\top\mathbf{y}(t)$, we obtain $\mathbf{C}(t)^\top(\mathbf{y}(t) - \mathbf{C}(t)\bm{\theta}) = \bm{0}$. Hence $\mathbf{w}:=\mathbf{y}(t)-\mathbf{C}(t)\bm{\theta} \in \text{Ker}(\mathbf{C}(t)^\top)$. 

Therefore, $\mathbf{w}\in\text{Im}(\mathbf{C}(t))^{\perp}$ by the identity $\text{Ker}(\mathbf{C}(t)^\top) = \text{Im}(\mathbf{C}(t))^\perp$, where $\text{Im}(\bullet)^\perp$ denotes the orthogonal complement of $\text{Im}(\bullet)$ \citep[Page 16]{horn}. On the other hand, since $\mathbf{y}(t)\in\text{Im}(\mathbf{C}(t))$, we conclude that $\mathbf{w}\in \text{Im}(\mathbf{C}(t))$. The only vector belonging simultaneously to $\text{Im}(\mathbf{C}(t))$ and its orthogonal complement is $\mathbf{w}=\bm{0}$.

Henceforth $\mathbf{y}(t)=\mathbf{C}(t)\bm{\theta}$, implying $\bm{\theta}\in\mathcal{W}$. This shows $\mathcal{W}'\subseteq\mathcal{W}$. Conversely, if $\bm{\theta}\in \mathcal{W}$, then pre-multiplying $\mathbf{y}(t)=\mathbf{C}(t)\bm{\theta}$ by $\mathbf{C}(t)^\top$ gives $\mathbf{y}'(t) = \mathbf{C}(t)^\top\mathbf{y}(t) = \mathbf{C}(t)^\top\mathbf{C}(t)\bm{\theta} = \mathbf{C}'(t)\bm{\theta}$, so $\bm{\theta}\in\mathcal{W}'$, i.e. $\mathcal{W}\subseteq\mathcal{W}'$. Therefore, $\mathcal{W}=\mathcal{W}'$.

\subsection{Auxiliary lemmas for Theorem \ref{the:gain_PE}}
\label{sec:convergence}

To facilitate the proof of Theorem \ref{the:gain_PE}, the following technical lemmas are introduced. 

\begin{lem}\label{lem:M_PE}
    Let $\mathbf{\overline{C}}(t)=\mathbf{C}^\top(t) \mathbf{C}(t)/N$. If $\mathbf{C}(t)$ in \eqref{Eq:MatCentralizado} is $\text{PE}(\alpha,T)$, for some $\alpha, T >0$, then $\mathbf{\overline{C}}(t)^2\in\text{PE}(\frac{\alpha^2}{T N^2}, T)$.
\end{lem}
 \begin{pf}
    By definition $\mathbf{\overline{C}}(t)=\frac{1}{N}\mathbf{C}(t)^\top\mathbf{C}(t)$, then $\int_t^{t+T} \mathbf{\overline{C}}(\tau)\text{d}\tau = \frac{1}{N} \int_t^{t+T} \mathbf{C}(\tau)^\top\mathbf{C}(\tau)\text{d}\tau \succeq \frac{\alpha}{N} \mathbf{I}_n.$
    Henceforth, for any $\mathbf{v}\in\mathbb{R}^n$ with $\|\mathbf{v}\|=1$ it follows that:
    \begin{equation*}
        \begin{aligned}
           &\left(\frac{\alpha}{N}\right)^2\leq \left[\mathbf{v}^\top \left( \int_t^{t+T} \mathbf{\overline{C}} (\tau)\text{d}\tau\right)\mathbf{v}\right]^2 
           \\&=
            \left[\int_t^{t+T} \mathbf{v}^\top\mathbf{\overline{C}}(\tau)\mathbf{v}\text{d}\tau\right]^2 
            \leq T\int_{t}^{t+T}(\mathbf{v}^\top \mathbf{\overline{C}}(\tau)\mathbf{v})^2\text{d}\tau\\
        \end{aligned}
    \end{equation*}
    where we used the Cauchy-Schwarz inequality with $g(\tau)=1$ and $h(\tau)=\mathbf{v}^\top\overline{\mathbf{C}}(\tau)\mathbf{v}$. Moreover, using again the Cauchy-Schwarz inequality:
    \begin{equation*}
    \begin{aligned}
            &(\mathbf{v}^\top \mathbf{\overline{C}}(\tau)\mathbf{v})^2 
             \leq \|\mathbf{v}\|^2\|\mathbf{\overline{C}}(\tau)\mathbf{v}\|^2 \\
            &=(\mathbf{\overline{C}}(\tau)\mathbf{v})^\top(\mathbf{\overline{C}}(\tau)\mathbf{v})=\mathbf{v}^\top \mathbf{\overline{C}}^\top(\tau) \mathbf{\overline{C}}(\tau) \mathbf{v}
            =\mathbf{v}^\top \mathbf{\overline{C}} (\tau)^2 \mathbf{v}\\ 
    \end{aligned}
    \end{equation*}
    Therefore $\frac{1}{T}\left(\frac{\alpha}{N}\right)^2\leq \int_t^{t+T}\mathbf{v}^\top \mathbf{\overline{C}} (\tau)^2 \mathbf{v}\text{d}\tau = \mathbf{v}^\top  \int_t^{t+T} \mathbf{\overline{C}} (\tau)^2 \text{d}\tau \mathbf{v}$.
    Recall that the previous inequality holds for all unit vectors $\mathbf{v}\in\mathbb{R}^n$. Then, $
    \frac{1}{T}\left(\frac{\alpha}{N}\right)^2\mathbf{I}_n \preceq \int_t^{t+T} \mathbf{\overline{C}} (\tau)^2 \text{d}\tau$ which completes the proof.
\end{pf}
 \begin{lem} \label{lem:Frobenius_error}
 Let the consensus error at agent $i$ be denoted as $\tilde{\mathbf{C}}_i(t) = \hat{\mathbf{C}}_i(t) - \overline{\mathbf{C}}(t)$. If the Assumption \ref{as:bounds}  holds then there exists a finite time $T_{\tilde{\mathbf{C}}}(\tilde{\mathbf{C}}_i(0))>0$ such that $$\| \tilde{\mathbf{C}}_i(t)\| \leq \frac{ n \gamma} {k \lambda_{\mathcal{G}}} \quad \forall t\geq T_{\tilde{C}}(\tilde{\mathbf{C}}_i(0)).$$
 \end{lem}
\begin{pf}
Stacking the local variables of all agents as $\mathbf{X}(t)= \text{col}(\mathbf{X}_i(t))_{i=1}^N,\quad \hat{\mathbf{C}}(t)=  \text{col}(\mathbf{\hat{C}}_i(t))_{i=1}^N$ and considering the average defined in \eqref{eq:DAC_limit}, we rewrite as $\overline{\mathbf{C}}(t)=\frac{1}{N}(\bm{1}^\top\otimes \mathbf{I}_n)\mathbf{C}'(t).$
Defining the stacked consensus error by $ \tilde{\mathbf{C}}(t)=\text{col}(\mathbf{\tilde{C}}_i(t))_{i=1}^N
=\hat{\mathbf{C}}(t)-(\bm{1}\otimes \mathbf{I}_n)\overline{\mathbf{C}}(t)$, where $\tilde{\mathbf{C}}_i(t)=(\mathbf{e}_i^\top\otimes \mathbf{I}_n)\tilde{\mathbf{C}}(t)$ is the error of agent $i$. The writing of \eqref{eq:Consensus_P}--\eqref{eq:Consensus_output} in vector form leads to:
\begin{equation}\label{eq:vector_form}
    \dot{\mathbf{X}}(t)=k(\mathbf{L}\otimes \mathbf{I}_n)\hat{\mathbf{C}}(t), 
\qquad
\hat{\mathbf{C}}(t)=\mathbf{C}'(t)-\mathbf{X}(t).
\end{equation}
Therefore, $\dot{\tilde{\mathbf{C}}}(t)
= \dot{\hat{\mathbf{C}}}(t)-(\bm{1}\otimes \mathbf{I}_n)\dot{\overline{\mathbf{C}}}(t)$
$$\begin{aligned}
\dot{\tilde{\mathbf{C}}}(t) &=\dot{\mathbf{C}}'(t)-\dot{\mathbf{X}}(t)-(\bm{1}\otimes \mathbf{I}_n)\left(\frac{1}{N}(\bm{1}^\top\otimes\mathbf{I}_n)\dot{{\mathbf{C}}}'(t)\right)  \\[4pt]
&= -k(\mathbf{L}\otimes \mathbf{I}_n)\tilde{\mathbf{C}}(t)+(\mathbf{H}\otimes \mathbf{I}_n)\dot{\mathbf{C}}'(t),
\end{aligned}
$$
where $\mathbf{H}=\mathbf{I}_N-\bm{1}\bm{1}^\top/N$.

A quadratic function is commonly employed as the Lyapunov candidate. Consistent with this, we employ the vectorization of $\tilde{\mathbf{C}}(t)$ and define
$V(\tilde{\mathbf{C}})=\text{vec}(\tilde{\mathbf{C}})^\top\text{vec}(\tilde{\mathbf{C}})$. By exploiting properties of the vectorization operator, this can be rewritten as $V(\tilde{\mathbf{C}})=\text{Tr}(\tilde{\mathbf{C}}^\top \tilde{\mathbf{C}})$, and its derivative is given by $\dot{V}(\tilde{\mathbf{C}}(t))=2 \text{Tr}(\tilde{\mathbf{C}}(t)^\top \dot{\tilde{\mathbf{C}}}(t))$.
Using the inequalities $|\text{Tr}(\mathbf{AB})|\leq \|\mathbf{A}\|_F\|\mathbf{B}\|_F$ and $||\mathbf{A}||_F \leq \sqrt{\text{rank} (\mathbf{A})} \|\mathbf{A}\|$ \citep[Pages 61-62]{petersen2008matrix}.
\begin{align*}
        &\dot{V}(\tilde{\mathbf{C}}(t))=\\ & 2\text{Tr}\left(-k \tilde{\mathbf{C}}(t)^\top (\mathbf{L} \otimes \mathbf{I}_n) \tilde{\mathbf{C}}(t)  + \tilde{\mathbf{C}}(t)^\top  (\mathbf{H \otimes I}_n)\mathbf{\dot{C}}' (t) \right) \\
        &\leq -2 k\lambda_\mathcal{G} \|\tilde{\mathbf{C}}(t) \|^2  + 2 n  \|\tilde{\mathbf{C}}(t)\| \|(\mathbf{H \otimes I}_n)\mathbf{\dot{C}}'(t) \|\\
        &=-2 \|\tilde{\mathbf{C}}(t)\| \left( k \lambda_{\mathcal{G}} \|\tilde{\mathbf{C}}(t) \|  -  n\|(\mathbf{H} \otimes \mathbf{I}_n) \mathbf{\dot{C}}'(t)\|  \right)
        \end{align*}
Thus, $\dot{V}(\tilde{\mathbf{C}}(t)) <0$ in the region where:  
$$
\|\tilde{\mathbf{C}}(t)\| \geq \frac{n\|(\mathbf{H} \otimes \mathbf{I}_n) \mathbf{\dot{C}}'(t)\|}{k \lambda_{\mathcal{G}}}.
$$
Hence, trajectories of $\mathbf{\tilde{C}}(t)$ will converge in finite time $T_{\tilde{\mathbf{C}}}(\tilde{\mathbf{C}}_i(0))>0$ to the region in which the following condition remains invariant:
$$
\begin{aligned}
\|\tilde{\mathbf{C}}(t)\|  &\leq \frac{n\|(\mathbf{H} \otimes \mathbf{I}_n) \mathbf{\dot{C}}'(t)\|}{k \lambda_{\mathcal{G}}}\\&\leq \frac{n\sup_{\tau\in[t,\infty)}\|(\mathbf{H} \otimes \mathbf{I}_n) \mathbf{\dot{C}}'(\tau)\|}{k \lambda_{\mathcal{G}}}.
\end{aligned}
$$
Due to Assumption \ref{as:bounds}$$
\begin{aligned}
\|\tilde{\mathbf{C}}(t)\|  &\leq \frac{n\gamma}{k \lambda_{\mathcal{G}}},
\end{aligned}
$$ 
completing the proof.
\end{pf}

{\begin{lem}\label{lem:regression_invariance}
Consider \eqref{eq:Consensus_P}--\eqref{eq:consensus_output_y} driven by the local surrogate signals \eqref{eq:cons_i}. Assume the standard DAC initialization $\mathbf{X}_i(0)=\bm{0}$ and $\mathbf{x}_i(0)=\bm{0}$ for all $i\in\{1,\dots,N\}$. Then, for all $t\ge 0$ and all $i\in\{1,\dots,N\}$,
$$
\hat{\mathbf{y}}_i(t)=\hat{\mathbf{C}}_i(t)\bm{\theta}.
$$
\end{lem}}

\begin{pf}
Define $\mathbf{r}_i(t)=\hat{\mathbf{y}}_i(t)-\hat{\mathbf{C}}_i(t)\bm{\theta}$. Using \eqref{eq:Consensus_output}--\eqref{eq:consensus_output_y} and \eqref{eq:cons_i}, we obtain
$$
\mathbf{r}_i(t)=\big(\mathbf{y}'_i(t)-\mathbf{x}_i(t)\big)-\big(\mathbf{C}'_i(t)-\mathbf{X}_i(t)\big)\bm{\theta}
=-\mathbf{x}_i(t)+\mathbf{X}_i(t)\bm{\theta}.
$$
Differentiating and using \eqref{eq:Consensus_P}--\eqref{eq:consensus_x} yields
$$
\begin{aligned}
\dot{\mathbf{r}}_i(t)&=-\dot{\mathbf{x}}_i(t)+\dot{\mathbf{X}}_i(t)\bm{\theta}
\\&=-k\sum_{j\in\mathcal{N}_i}\big(\hat{\mathbf{y}}_i(t)-\hat{\mathbf{y}}_j(t)\big)
+k\sum_{j\in\mathcal{N}_i}\big(\hat{\mathbf{C}}_i(t)-\hat{\mathbf{C}}_j(t)\big)\bm{\theta}
\\&=-k\sum_{j\in\mathcal{N}_i}\big(\mathbf{r}_i(t)-\mathbf{r}_j(t)\big).
\end{aligned}
$$
Equivalently, in stacked form $\mathbf{r}(t)=\mathrm{col}(\mathbf{r}_i(t))_{i=1}^N$ satisfies
$$
\dot{\mathbf{r}}(t)=-k(\mathbf{L}\otimes\mathbf{I}_n)\mathbf{r}(t).
$$
Moreover, $\mathbf{r}(0)=\bm{0}$ since $\mathbf{X}_i(0)=\bm{0}$ and $\mathbf{x}_i(0)=\bm{0}$ for all $i$. Hence, $\mathbf{r}(t)=\bm{0}$ for all $t\ge 0$, which implies $\hat{\mathbf{y}}_i(t)=\hat{\mathbf{C}}_i(t)\bm{\theta}$ for all $t\ge 0$.
\end{pf}
\subsection{Proof of Theorem \ref{the:gain_PE}}\label{sec:proof:theorem}
Given \eqref{eq:alg_SGE}, define the estimate error of agent $i$ as
$\tilde{\bm{\theta}}_i=\hat{\bm{\theta}}_i-\bm{\theta}$.
By Lemma~\ref{lem:regression_invariance}, under the standard DAC initialization
$\mathbf X_i(0)=\bm 0$ and $\mathbf x_i(0)=\bm 0$, the consensus outputs satisfy
$\hat{\mathbf y}_i(t)=\hat{\mathbf C}_i(t)\bm{\theta}$ for all $t\ge 0$.
Moreover, since $\mathbf C_i'(t)=\mathbf C_i(t)^\top\mathbf C_i(t)$ is symmetric and
$\mathbf X_i(0)=\mathbf{0}$ is symmetric, the DAC dynamics preserve symmetry, and thus
$\hat{\mathbf C}_i(t)$ is symmetric for all $t\ge 0$.
Substituting into \eqref{eq:alg_SGE} yields
$$
\begin{aligned}
\dot{\hat{\bm{\theta}}}_i(t)
&=\bm{\Gamma}_i \hat{\mathbf{C}}_i(t)^\top\big(\hat{\mathbf{y}}_i(t) - \hat{\mathbf{C}}_i(t)\boldsymbol{\hat{\theta}}_i(t)\big)\\&=\bm{\Gamma}_i\,\hat{\mathbf C}_i(t)\big(\hat{\mathbf C}_i(t)\bm{\theta}-\hat{\mathbf C}_i(t)\hat{\bm{\theta}}_i(t)\big)
\\&=\bm{\Gamma}_i\,\hat{\mathbf C}_i(t)^2\big(\bm{\theta}-\hat{\bm{\theta}}_i(t)\big).
\end{aligned}
$$

which gives the error dynamics
\begin{equation}\label{eq:parameter_error_sge}
\dot{\tilde{\bm{\theta}}}_i(t)
=-\bm{\Gamma}_i\,\hat{\mathbf C}_i(t)^2\,\tilde{\bm{\theta}}_i(t).
\end{equation}

As indicated in Proposition \ref{prop:error_PE}, if $\mathbf{\hat{C}}_i(t)$ is PE, then $\tilde{\bm{\theta}}_i \to 0$, consequently $\hat{\boldsymbol{\theta}}_i \to \boldsymbol{\theta}$. Therefore, the rest of the proof shows that $\mathbf{\hat{C}}_i(t)$ is PE, provided the consensus gains complies with  \eqref{eq:gain_k}.

The consensus output at agent $i$ is defined as $\hat{\mathbf{C}}_i(t)=\tilde{\mathbf{C}}_i(t)+\overline{\mathbf{C}}(t)$. The integral to be evaluated is:
 \begin{equation*}
     \begin{split}
        &\int_t^{t+T} \mathbf{\hat{C}}_i(\tau)^2 \text{d}\tau=
        \int_t^{t+T}\left(\mathbf{\overline{C}}(\tau)+\tilde{\mathbf{C}}_i(\tau)\right)^2\text{d}\tau\\& \quad =\int_t^{t+T}\left( \mathbf{\overline{C}} (\tau)^2 + \tilde{\mathbf{C}}_i(\tau)^2 + \mathbf{\overline{C}}(\tau) \tilde{\mathbf{C}}_i(\tau) + \tilde{\mathbf{C}}_i(\tau) \mathbf{\overline{C}}(\tau) \right)\text{d}\tau
     \end{split}
 \end{equation*}
 considering $\int_t^{t+T}\tilde{\mathbf{C}}_i (\tau)^2 \text{d}\tau \succeq \boldsymbol{0}$ and according to Lemma \ref{lem:M_PE} it follows that $\int_t^{t+T}\left( \mathbf{\overline{C}}(\tau)\right)^2\text{d}\tau \succeq \frac{\alpha^2}{T N^2} \mathbf{I}_n$. We use the Euclidean norm on the non-square terms of the integral.
  \begin{align}
    \|\mathbf{\overline{C}}(t) \tilde{\mathbf{C}}_i(t) +\tilde{\mathbf{C}}_i(t) \mathbf{\overline{C}}(t)\| &\leq \|\mathbf{\overline{C}}(t) \tilde{\mathbf{C}}_i(t) \|+ \|\tilde{\mathbf{C}}_i(t) \mathbf{\overline{C}}(t)\| \nonumber \\
         &\leq 2\|\mathbf{\overline{C}}(t)\| \|\tilde{\mathbf{C}}_i(t) \|\label{eq:bound_int}
 \end{align}

 By assumption $\|\mathbf{\overline{C}}(t)\|\leq \beta$ and according to Lemma \ref{lem:Frobenius_error} $\|\tilde{\mathbf{C}}_i(t)\| \leq \frac{n \gamma}{\lambda_{\mathcal{G}} k}$ then $2\|\mathbf{\overline{C}}(t)\| \|\tilde{\mathbf{C}}_i(t) \| \leq 2\beta\frac{n\gamma}{\lambda_{\mathcal{G}} k}$. Therefore
 \begin{equation*}
    \begin{split}
        \int_t^{t+T}\!(\mathbf{\hat{C}}_i(\tau))^2\text{d}\tau & \succeq \int_t^{t+T}\! \left(\mathbf{\overline{C}}^2(\tau) - 2\beta\left(\frac{n \gamma}{\lambda_{\mathcal{G}}k}\right)\mathbf{I}_n \right) \text{d}\tau\\ &\succeq \left(\frac{\alpha^2}{T N^2} - 2\beta\left(\frac{n \gamma}{\lambda_{\mathcal{G}}k}\right) T \right) \mathbf{I}_n = \varphi \mathbf{I}_n
    \end{split}
\end{equation*}

Consequently, by assumption $k > \frac{2n\beta \gamma T^2 N^2}{\lambda_{\mathcal{G}} \alpha^2}$ making $\varphi > 0$. As a result, $\mathbf{\hat{C}}_i(t)$ is PE, which completes the proof.

\subsection{Proof of Corollary~\ref{cor:quantization}}\label{proof:quant}
For each agent $i$,
\[
\mathcal{Q}(\hat{\mathbf C}_i(t))=\hat{\mathbf C}_i(t)+\bm{\varepsilon}_i(t),
\qquad
\mathcal{Q}(\hat{\mathbf y}_i(t))=\hat{\mathbf y}_i(t)+\bm{\varepsilon}_i^{y}(t),
\]
with
\[
\|\bm{\varepsilon}_i(t)\|\le n\,\varepsilon,
\qquad
\|\bm{\varepsilon}_i^{y}(t)\|\le \sqrt{n}\,\varepsilon,
\quad \forall t\ge 0.
\]
Stacking $\bm{\varepsilon}_i$ as $\bm{\varepsilon}=\col(\bm{\varepsilon}_i)_{i=1}^N$, with $\|\bm{\varepsilon}(t)\| \leq n \sqrt{N}\varepsilon \ $ and $\bm{\varepsilon}^y=\col(\bm{\varepsilon}^y_i)_{i=1}^N$ with $\|\bm{\varepsilon}^y(t)\| \leq \sqrt{nN} \varepsilon, \ \forall t\geq 0$, the vector form of \eqref{eq:Consensus_P:quant}--\eqref{eq:consensus_output_y} yields
\[
\dot{\tilde{\mathbf C}}(t)= -k(\mathbf L\otimes \mathbf I_n)\tilde{\mathbf C}(t)+(\mathbf H\otimes \mathbf I_n)\dot{\mathbf C}'(t)-k(\mathbf L\otimes \mathbf I_n)\bm{\varepsilon}(t),
\]
where $\tilde{\mathbf C}(t)=\hat{\mathbf C}(t)-(\bm{1}\otimes \mathbf I_n)\overline{\mathbf C}(t)$ and $\mathbf H=\mathbf I_N-\bm{1}\bm{1}^\top/N$. Using the same Lyapunov function $V(\tilde{\mathbf C})=\mathrm{Tr}(\tilde{\mathbf C}^\top\tilde{\mathbf C})$ as in Lemma~\ref{lem:Frobenius_error}, the additional term $-k(\mathbf L\otimes \mathbf I_n)\bm{\varepsilon}(t)$ acts as a bounded disturbance, which implies that $\tilde{\mathbf C}(t)$ is uniformly ultimately bounded. In particular, there exists a finite time $T_q(\tilde{\mathbf{C}}_i(0))>0$ such that, for all $t\ge T_q(\tilde{\mathbf{C}}_i(0))$ and all $i$,
\begin{equation}\label{eq:error_C_quant}
\|\tilde{\mathbf C}_i(t)\|
\le
\frac{n\gamma}{k\lambda_{\mathcal G}}
+
\frac{\varepsilon n^2 \sqrt{N} \lambda_{\max}(\mathbf L)}{\lambda_{\mathcal G}}\,=:b(\varepsilon).
\end{equation}
Therefore, for $t\ge T_q(\tilde{\mathbf{C}}_i(0))$,
$
\hat{\mathbf C}_i(t)=\overline{\mathbf C}(t)+\tilde{\mathbf C}_i(t)
$
is a bounded perturbation of $\overline{\mathbf C}(t)$. Since $\overline{\mathbf C}(t)^2\in\mathrm{PE}\!\left(\frac{\alpha^2}{TN^2},T\right)$ by Lemma~\ref{lem:M_PE}, the same integral lower-bound argument used in the proof of Theorem~\ref{the:gain_PE} gives, for all $t\ge T_q(\tilde{\mathbf{C}}_i(0))$,
\begin{align}\label{eq:qua_PE}
&\int_t^{t+T}\hat{\mathbf C}_i(\tau)^2\,\mathrm d\tau \succeq  \varphi \mathbf{I}_n,
\end{align}
where 
$$
\varphi =  \frac{\alpha^2}{TN^2} - 2 \beta T \Bigg(\frac{n\gamma}{k\lambda_{\mathcal G}} +
\frac{\varepsilon n^2\sqrt{N}\lambda_{\max}(\mathbf L)}{\lambda_{\mathcal G}}\Bigg).
$$
Under the condition stated in Corollary~\ref{cor:quantization}, $\varphi$ is strictly positive, and hence $\hat{\mathbf C}_i(t)$ is persistently exciting. 

Due to quantization, the exact relation $\hat{\mathbf y}_i(t)=\hat{\mathbf C}_i(t)\bm{\theta}$ no longer holds. In fact, 
the consensus relation yields
\[
\hat{\mathbf y}_i(t)
=\hat{\mathbf C}_i(t)\bm{\theta}
+\mathbf r_i(t).
\]
Similarly to the proof of Lemma \ref{lem:regression_invariance}, using \eqref{eq:Consensus_output}--\eqref{eq:consensus_output_y}
\begin{align*}
    \mathbf{r}_i(t)=\mathbf{X}_i(t)\bm{\theta}-\mathbf{x}_i(t),
\end{align*}
differentiating and using \eqref{eq:Consensus_P:quant}--\eqref{eq:Consensus_x:quant}
\begin{align*}
    \dot{\mathbf{r}}_i(t)=&-\dot{\mathbf{x}}_i(t)+\dot{\mathbf{X}}_i(t)\bm{\theta}\\
    =&-k\!\!\sum_{j\in\mathcal{N}_{i}}\!
    \big(\mathcal{Q}(\hat{\mathbf{y}}_i(t))-\mathcal{Q}(\hat{\mathbf{y}}_j(t))\big)\\&+k\!\!\sum_{j\in\mathcal{N}_{i}}\!
    \big(\mathcal{Q}(\hat{\mathbf{C}}_i(t))-\mathcal{Q}(\hat{\mathbf{C}}_j(t))\big)\bm{\theta},    
\end{align*}
stacking $\mathbf{r}(t)=\col(\mathbf{r}_i(t))_{i=1}^N$, the vector form is
\begin{align*}
    \dot{\mathbf{r}}(t)&=k(\mathbf{L} \otimes \mathbf{I}_n)\Big(\hat{\mathbf{C}}(t)\bm{\theta} - \hat{\mathbf{y}}(t)-\bm{\varepsilon}(t)\bm{\theta} + \bm{\varepsilon}^y(t)\Big)\\
    &=-k(\mathbf{L}\otimes\mathbf{I}_n)\mathbf{r}(t)+k(\mathbf{L}\otimes\mathbf{I}_n)(\bm{\varepsilon}^y(t) - \bm{\varepsilon}(t))
\end{align*}
Define a candidate Lyapunov function as $V(\mathbf{r})=\frac{1}{2}\mathbf{r}(t)^\top \mathbf{r}(t)$, with
$$\dot{V}(\mathbf{r}(t)) =-k \mathbf{r}(t)^\top(\mathbf{L} \otimes \mathbf{I}_n)\Big(\mathbf{r}(t) - \bm{\varepsilon}^y(t) + \bm{\varepsilon}(t) \bm{\theta}\Big).$$
Since $\mathbf{r}(0)=\bm{0}$ and $(\bm{1}^\top \otimes\mathbf{I}_n)\dot{\mathbf{r}}(t)=\bm{0}$, then  $(\bm{1}^\top \otimes\mathbf{I}_n)\mathbf{r}(t)=\bm{0}$ for all $t$, and therefore we obtain $$k\mathbf{r}(t)^\top(\mathbf{L}\otimes \mathbf{I}_n)\mathbf{r}(t) \geq k \lambda_\G \|\mathbf{r}(t)\|^2.$$ Hence
\begin{align*}
    \dot{V}(\mathbf{r}(t)) 
    \leq &-k \|\mathbf{r}(t)\|^2\lambda_\G \\ &+k\|\mathbf{r}(t)\|\lambda_{\max}(\mathbf{L}) \big( \|\bm{\varepsilon}^y(t)\| + \|\bm{\varepsilon} \| \|\bm{\theta}\|\big)\\
    \leq & -k \|\mathbf{r}(t)\|^2\lambda_\G \\ &+k\|\mathbf{r}(t)\|\lambda_{\max}(\mathbf{L}) \big(\sqrt{nN} \varepsilon +n\sqrt{N} \varepsilon \|\bm{\theta}\| \big).
\end{align*}
Substituting $2V(\mathbf{r}(t))=\|\mathbf{r}(t)\|^2$,
\begin{align*}
    \dot{V}(\mathbf{r}(t))\leq&-2k\lambda_\G V(\mathbf{r}(t)) \\ &+k \sqrt{2nNV(\mathbf{r}(t))} \lambda_{\max} (\mathbf{L}) \varepsilon \big(\sqrt{n}\|\bm{\theta}\|+1)\\
    \leq& -k \sqrt{V(\mathbf{r}(t))} \Big(2\lambda_\G  \sqrt{V(\mathbf{r}(t))} \\ &- \varepsilon \sqrt{2nN} \lambda_{\max}(\mathbf{L})(\sqrt{n}\|\bm{\theta}\| +1)\Big)
\end{align*}

Thus, $V(\mathbf{r}(t)) <0$ in the region where $$\sqrt{V(\mathbf{r}(t))} \geq \frac{\varepsilon \sqrt{2nN}\lambda_{\max}(\mathbf{L})(\sqrt{n}\|\bm{\theta}\| +1)}{2 \lambda_\G}.$$

Then $\mathbf{r}(t)$ will converge in finite time $T_\mathbf{r}(\mathbf{r}(0))>0$ to the invariant region:
\begin{align}
    \frac{1}{\sqrt{2}}\|\mathbf{r}(t)\| = \sqrt{V(\mathbf{r}(t))}\leq & \frac{\varepsilon\sqrt{2nN} \lambda_{\max}(\mathbf{L})  (\sqrt{n} \|\bm{\theta} \| +1)}{2 \lambda_\G} \nonumber\\
    \|\mathbf{r}_i(t)\| \leq \|\mathbf{r}(t)\| \leq &r,
\end{align}
where 
$$
r(\varepsilon) := \frac{\varepsilon \sqrt{nN}\lambda_{\max}(\mathbf{L})(\sqrt{n}\|\bm{\theta}\| +1)}{ \lambda_\G}.
$$
The estimation error is defined as $\tilde{\bm{\theta}}_i=\hat{\bm{\theta}}_i-\bm{\theta}_i$. Considering \eqref{eq:alg_SGE}:
\begin{equation}
\begin{aligned}
\label{eq:error:lyap}
    \dot{\tilde{\bm{\theta}}}_i&=\bm{\Gamma}_i \hat{\mathbf{C}}_i(t)^\top \Big( \hat{\mathbf C}_i(t) \bm{\theta} + \mathbf r_i(t)-\hat{\mathbf C}_i(t) \hat{\bm{\theta}}_i(t)
\Big)\\
&= - \bm{\Gamma}_i \hat{\mathbf{C}}_i(t)^\top \hat{\mathbf C}_i(t) \tilde{\bm{\theta}}_i(t)
 + \bm{\Gamma}_i\hat{\mathbf{C}}_i(t)^\top \mathbf r_i(t)
 \end{aligned}
\end{equation}
Let $\|\bm{\Gamma}_i\| \leq a_i$ and $\|\hat{\mathbf{C}}_i(t)\| \leq \|\tilde{\mathbf{C}}_i(t)\| + \|\Bar{\mathbf{C}}(t)\|$. By Assumption~\ref{as:bounds}, in conjunction with \eqref{eq:error_C_quant}
it follows that $\|\hat{\mathbf{C}}_i(t)\|\leq b+\beta$ after a finite time. Finally, take the strong Lyapunov function candidate $V(\tilde{\bm{\theta}}_i)$ from \citep{rueda2021strong} which ensures $\dot{V}(\tilde{\bm{\theta}}_i(t))\leq -cV(\tilde{\bm{\theta}}_i(t))$ with $c>0$ for \eqref{eq:error:lyap} under PE on the nominal case, when the disturbance satisfies $\hat{\mathbf{C}}_i(t)\mathbf{r}_i(t)=\bm{0}$ as a result of \citep[Theorem 1]{rueda2021strong}. For non zero disturbance, it follows that $\|\hat{\mathbf{C}}_i(t)\mathbf{r}_i(t)\|\leq (b(\varepsilon)+\beta)r(\varepsilon)$ after a finite time, with the disturbance bound $(b(\varepsilon)+\beta)r(\varepsilon)$ being a strictly increasing function of $\varepsilon$ with $(b(0)+\beta)r(0)=0$. These arguments are used to apply \citep[Theorem 4.19]{khalil2002nonlinear} to obtain input-to-state stability for \eqref{eq:error:lyap}, concluding the proof. 

\subsection{Proof of Corollary \ref{cor:switched}}
The practical convergence of the Switched Dynamic Average Consensus  block in \eqref{alg:distributed_estimator:switched} follows the same arguments as in Lemma \ref{lem:Frobenius_error} and Corollary \ref{cor:quantization}, since the Lyapunov function $V(\tilde{\mathbf{C}})=\text{Tr}(\tilde{\mathbf{C}}^\top \tilde{\mathbf{C}})$ for the error system is graph-independent. Thus, it constitutes a common Lyapunov function for all graphs in the family $\{\mathcal{G}_1\dots,\mathcal{G}_q\}$, ensuring convergence of the resulting switched error system \citep{commonlyap}. Henceforth, similarly to \eqref{eq:error_C_quant} the error converges in finite time to:
$$
\|\tilde{\mathbf C}_i(t)\|
\le
\frac{n\gamma}{k\lambda_{\mathcal G,m}}
+
\frac{\varepsilon n^2 \sqrt{N} \lambda_{\mathcal G, M}}{\lambda_{\mathcal G,m}}
$$
uniformly for all graphs in the switching topology family. From this point on, the rest of the proof follows in the same way as in Corollary \ref{cor:quantization}.

\subsection{Proof of Corollary \ref{cor:DREM_L2}}\label{sec:proof:DREM_L2}
    Given \eqref{eq:alg_Drem_C}--\eqref{eq:alg_drem_y}, then $\mathbf{y}_i^f (t) =\mathbf{C}_i^f (t)\boldsymbol{\theta}$, premultiplaying by $[\text{adj}(\mathbf{C}_i^f(t)^\top \mathbf{C}_i^f (t)) \mathbf{C}_i^f(t)^\top]$, we obtain a scalar equation by each parameter as $\mathbf{Y}_i(t)= \phi_i(t) \boldsymbol{\theta}$, where $\phi_i(t)$ and $\mathbf{Y}_i(t)$ are defined in \eqref{eq:alg_DREM_Yg}--\eqref{eq:alg_phi}, respectively. The parameter estimation is defined in \eqref{eq:alg_DREM_th}. The time-varying dynamical of the estimation error $\boldsymbol{\tilde{\theta}}_{i}=\boldsymbol{\hat{\theta}}_{i} - \boldsymbol{\theta}$ is given by
\begin{equation}\label{eq:error_est_drem}
    \boldsymbol{\dot{\tilde{\theta}}}_{i}=-\boldsymbol{\Gamma}_{i} \phi_i(t)^2 \boldsymbol{\tilde{\theta}}_{i}.
\end{equation}

It follows that to ensure that the estimation error converges exponentially to zero, it suffices that 
\begin{equation}\label{eq:condition_DREM}
    \phi_i(t) \notin \mathcal{L}_2,    
\end{equation} therefore, $\hat{\boldsymbol{\theta}}_i$ in \eqref{eq:alg_DREM_th} converges to the true value $\boldsymbol{\theta}$, completing the proof.
 \subsection{Proof of Corollary \ref{cor:DREM}}\label{sec:proof:DREM}
The error in \eqref{eq:error_est_drem}, when $\phi_i(t)=\det (\hat{\mathbf{C}}_i(t))$, represents a set of scalar equations, one for each parameter, such that convergence to zero depends on the condition \eqref{eq:condition_DREM} being satisfied. Consequently, if $\det(\hat{\mathbf{C}}(t)) \notin \mathcal{L}_2$, then $\hat{\boldsymbol{\theta}}_i$ in \eqref{eq:alg_DREM_cor}, converges to true value $\boldsymbol{\theta}$.

    Furthermore, \citep[Proposition 3]{ortega2020DREM_PE}  states that if a matrix is PE, then its determinant does not belong to $\mathcal{L}_2$, so when the matrix $\mathbf{\hat{C}}_i(t)$ is persistently exciting (PE), then $\phi_i(t) = \det\! \left(\mathbf{\hat{C}}_i(t) \right)\notin \mathcal{L}_2$, which guaranties the convergence of the Algorithm \eqref{alg:DREM_NotFilter}.

\end{document}